\documentclass[journal,11pt,onecolumn]{IEEEtran}
\usepackage{amsfonts, amsthm}
\usepackage{amsmath, color}  
\usepackage{graphicx}
\usepackage{cite}

\newtheorem{theorem}{Theorem}
\newtheorem{definition}{Definition}

\newtheorem{lemma}{Lemma}

\newtheorem{corollary}{Corollary}
\newtheorem{remark}{Remark}
\newtheorem{example}{Example}

\newcommand{\E}{\mathbb{E}}
\newcommand{\Po}{\mathsf{Poisson}}
\allowdisplaybreaks
\title{Capacity of Diffusion based Molecular Communication Networks over LTI-Poisson Channels}
\author{\IEEEauthorblockN{Hamidreza Arjmandi\IEEEauthorrefmark{1},
Gholamali Aminian\IEEEauthorrefmark{1}, Amin Gohari\IEEEauthorrefmark{1}, Masoumeh Nasiri Kenari\IEEEauthorrefmark{1} and Urbashi Mitra\IEEEauthorrefmark{2}}\\
\IEEEauthorblockA{\IEEEauthorrefmark{1} Sharif University of Technology              \IEEEauthorrefmark{2}University of Southern California (USC)}}
\begin{document}
\maketitle
\begin{abstract}
In this paper, the capacity of a diffusion based molecular communication network under the model of a Linear Time Invarient-Poisson (LTI-Poisson) channel is studied. Introduced  in the context of molecular communication, the LTI-Poisson model is a natural extension of the conventional memoryless Poisson channel to include memory. Exploiting prior art on linear ISI channels, a computable finite-letter characterization of the capacity of single-hop  LTI-Poisson networks is provided. Then, the problem of finding more explicit bounds on the capacity is  examined, where lower and upper bounds for the point to point case are provided. Furthermore, an approach for bounding mutual information in the low SNR regime using the symmetrized KL divergence is introduced and its applicability to Poisson channels is shown. To best of our knowledge, the first non-trivial upper bound on the capacity of Poisson channel with a maximum transmission constraint in the low SNR regime is found. Numerical results show that the proposed upper bound is of the same order as the capacity in the low SNR regime.
\end{abstract}

\section{Introduction}
Design of effective communication schemes via nano-machines is motivated by successes in the development of these units. Many applications in biomedical, industrial and environmentally engineered systems are envisioned for interconnected nano-networks \cite{Akyl2011}, \cite{Nakano2012}. Inspired by biological systems, in this work, we consider a \emph{molecular communication} (MC) system via diffusion \cite{Pierboon2010}, \cite{Eckford2007}. In a diffusion based MC system, information is encoded into the concentration, type, or release time of the molecules diffused into the medium by a nano-transmitter. Molecules travel from the transmitter(s) to receiver(s) via a Brownian motion mechanism, that can be with or without drift (see \cite{Book}).

To understand the fundamental transmission capacity of diffusion based MC systems, one has to deal with the unique features of the diffusion channel, including the intersymbol and co-channel interferences which are due to the gradual diffusion process and the fact that molecules released from a transmitter can reach the receiver after a long delay. Interference of diffusion based MC has been studied in \cite{Mahfuz2011,B1,B2,Pierobon2012,Atakan2012,Pierobon2014} for different modulation schemes. Further, several coding schemes are proposed to mitigate ISI in diffusion based MC  \cite{Leeson2012, Ko2012, Shih2012}. 

One of the first papers to address the need for a mathematical analysis of the capacity of MC systems is \cite{Alfano2006}.  In \cite{CapacityN2,Akan2008,Akan2007,Ein2011,Liu2013,Kuran2012,Airfler2011,CapacityN1,CapacityN15}, the authors consider special transmission strategies (binary or quaternary) for diffusion based MC and analyze the achievable information rates.  Some of these works also consider the interference of the last transmitted symbol. In \cite{Eckford2008,Khormuji2011,Eckford2012J,Eckford2012C} the authors study the achievable information rates in MC for a timing channel in which information is encoded in the time of release of the molecules. In \cite{Pierobon2013}, the capacity of diffusion based MC is studied under a large scale channel model. The paper \cite{Hsieh2013} proposes a general model for the diffusion based MC channel and mathematical formalisms for studying capacity, but does not provide explicit capacity formulas.

A point to point diffusion channel (with or without drift in a uniform or non-uniform medium) can be modeled as a state dependent channel, with the state modeling the number or density of molecules across the environment. One can observe that this diffusion channel is \emph{indecomposable} \cite[p. 105]{Gallager} as the initial state diffuses away over time. Therefore, one can write out the formula to characterize capacity in a way that is computable, although the results of \cite{Gallager} only provide a \emph{finite}-letter form. This is a simple observation, but to best of our knowledge, the literature on molecular communication does not point this out. 

The capacity of networks with memory has been the subject of numerous studies in information theory. These channels are generally modeled via finite state or general channel models. In the class of finite state channels, linear ISI channels have received particular attention, e.g. see \cite{Massey, Verdu}. There has been relatively less work on multi-terminal networks with memory. Limiting (uncomputable) characterizations of single-hop networks have been provided in \cite{Relayless}. The capacity region of certain networks with a MAC architecture have been found by Verd\'{u} in \cite{Verdu}. Dabora and Goldsmith find the capacity of degraded finite-state broadcast channel \cite{Dabora} where they face the  superposition coding aspect of the region that does not exist in the point to point and the MAC counterparts.

In this paper, we prove several capacity results for the diffusion model of \cite{Arjmandi2013, Mosayebi2014} that models a stationary point to point diffusion channel in an arbitrary environment. We call this model the \emph{LTI-Poisson} model for reasons that will become clear later. The same model can be also used for the bacterial filament problem in \cite{Nicola2014} for some special cases.  The LTI-Poisson model can be understood as a generalization of the classical memoryless Poisson channel. The Poisson channel has applications in optical communications and has been the subject of many studies.  
Therefore, the LTI-Poisson model relates to two bodies of literature in information theory: networks with memory and memoryless Poisson channels. A common point in both literatures is an attempt to find easy-to-compute expressions for the capacity (e.g. see \cite{Lapidoth2011, LapidothMoser2009}). In particular, both literatures exploit Topsoe's upper bound on mutual information  \cite{Topsoe, AllertonUpper} (see Remark \ref{remarkTopsoe} for details of this inequality). One of the goals of this paper is to develop a new approach to bound mutual information from above, in the low SNR regime.  

 In this paper, after generalizing the LTI-Poisson model to a single-hop multi-terminal setting, we make the following contributions:
\begin{itemize}
\item We provide a computable characterization of the capacity region (Theorem \ref{thm1}) under the LTI-Poisson model. This result follows from classical ideas used in studying the capacity of linear ISI channels \cite{Massey, Verdu, Goldsmith}, as we show that the model of  \cite{Arjmandi2013, Mosayebi2014}  falls into the same general category as linear ISI channels.
\item Next, we consider the special case of a point to point channel. Using Theorem \ref{thm1}, we develop some lower and upper bounds on the capacity (Theorem \ref{thm2}). Further we provide some useful lemmas and numerical simulations. Numerical results provide an estimate of how fast the multi-letter characterization approaches the capacity region for a Poisson channel, as we vary the channel parameters.
\item Finally, we develop a new (easy to compute) upper bound on mutual information using symmetrized Kullback-Leibler divergence (KL divergence). Based on numerical evidence, we believe that this upper bound works well for channels with small capacity (which can occur in MC systems). 
\end{itemize}

 Throughout this paper all the logarithms are in base $e$. 

This paper is organized as follows: in Section \ref{SystemModel}, we review the LTI-Poisson model of  \cite{Arjmandi2013} and extend it to a multi-terminal setting. Here, we define a class of memory limited networks that generalizes both the linear ISI channel and the LTI-Poisson model. Section 
\ref{capacity} includes our general capacity results. Section \ref{symm} contains a new upper bound technique for mutual information that can be used to find easy-to-compute approximations of a capacity region. Finally, Section \ref{LTIpoisson} applies our results on the memory limited networks to the LTI-Poisson model. Key proofs are given in the appendices.

\section{System Model}
\label{SystemModel}

We review the appropriate molecular Poisson channel model of \cite{Arjmandi2013} for point to point communication. Time is assumed to be divided into equal time-slots, during which a transmission and a reception occurs. The transmitter opens the outlet of a molecule storage for a short period of time at the beginning of each time slot. The input $X_i$, at time slot $i$, controls the size of the outlet of a molecule storage, from which molecules can flee. Should the transmitter choose $X_i$ as the size of its outlet, the number of molecules that are released into the environment will follow a Poisson distribution with parameter $X_i$. Each of these molecules hit the receiver in next $k$-th time slot with probability $p_k$, $k=0,1,2,\cdots$, upon which the receiver absorbs the molecule. From the thinning property of the Poisson distribution (random selection of Poisson points,\cite{thinning}), we know that the number of molecules due to transmission $X_i$ that are received in the subsequent time slots, follow a Poisson distribution. More precisely, assuming that $X_i$ is the input to a channel at time slot $i$ for $i\geq 0$, the output at time instance $i$ is 
\begin{align}Y_i\sim\Po(\lambda_0+\sum_{j=0}^ip_{j}X_{i-j}),\qquad\forall i\geq 0\label{eqn1nn}\end{align}
where $\lambda_0$ is the background noise and $\textbf{p}=(p_0, p_1, p_2, \cdots)$ is a sequence of non-negative ``hitting probabilities" satisfying $\sum_{i}p_i<\infty$. Equation \eqref{eqn1nn} only considers the marginal pmf of the output $Y_i$ given the channel inputs. Indeed conditioned on inputs $X_{1:n}$, the outputs are mutually independent and hence \eqref{eqn1nn} is a full description of the channel statistics. In other words,
\begin{align}p(y_{1:n}|x_{1:n})&=\prod_{i=1}^np(y_i|x_{1:i})\\
&=\prod_{i=1}^n
e^{-(\lambda_0+\sum_{j=0}^ix_{i-j}p_{j})}\frac{(\lambda_0+\sum_{j=0}^ix_{i-j}p_{j})^{y_i}}{y_i!}.
\end{align}
This equation follows from the thinning property of the Poisson distribution, as given a sequence of transmissions, the molecules received from each transmission are independent across various time slots.

The sequence of hitting probabilities fully captures the impact of  the geometry of the communication medium, its possible non-uniformity, the distance between transmitter and receiver,  and drift. For the special case of $p_0=1$ and $p_i=0$ for $i\geq 1$, we get $Y_i\sim\Po(\lambda_0+X_i),$ which is the well-known, discrete time, Poisson channel. As with the classical  Poisson channel, we assume the following constraints on the input codewords of length $n$: $X_i\geq 0$, average input constraint $\sum_{i=1}^nX_i/n\leq \alpha$ and possibly a constraint on the maximum value of $X_i$, $X_i\leq \mathsf{A}$. 

The above model is justified in \cite{Arjmandi2013} and \cite{Mosayebi2014}. We provide a simple observation: the summation $\sum_{j=0}^ip_{j}X_{i-j}$ is the convolution of the sequence $(X_0, X_1, \cdots)$ with the sequence $\textbf{p}=(p_0, p_1, p_2, \cdots)$. Therefore $Y_i$ can be understood as a cascade of an LTI system (with impulse response $\textbf{p}$) and a memoryless Poisson channel.  This is depicted in Fig. \ref{Figure1}. For this reason, we call this channel an \emph{LTI-Poisson channel}.

\begin{figure*}
\begin{center}
\includegraphics[scale=0.3,angle=0]{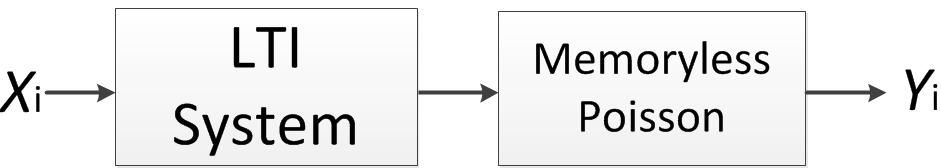}
\end{center}
\caption{An LTI-Poisson Channel: a model for the diffusion channel}
\label{Figure1}
\end{figure*}

More generally, consider a network where each node is either a transmitter or a receiver. In particular, assume that we have $s$ transmitters and $d$ receivers. For any transmitter node $i$ and receiver node $j$, there is a message of rate $R_{ij}$ to be transmitted. These types of networks include broadcast, MAC and interference channel architectures, but not the relay channel for instance. See Part II of \cite{ElGamalKim} for classical results on single-hop networks.

The set of transmissions by user $j$ in times $1,2, \cdots, n$ will be denoted by $x_{j1}, x_{j2}, \cdots, x_{jn}$. Let 
$\textbf{x}_i=(x_{1i}, \cdots, x_{si})$ be the set of inputs for the $s$ transmitters at time $i$. Similarly,  $\textbf{y}_i=(y_{1i}, \cdots, y_{di})$  is the set of outputs of the $d$ receivers at time $i$. We can straightforwardly extend the derivations of  \cite{Arjmandi2013} to determine the relationship between the input and output and show that this system can be modeled via an LTI-Poisson network given in Fig. \ref{Figure1.5}. To see this, let $p_{l,j,k}$ for some $1\leq l\leq s$, $1\leq j\leq d$ and $k\geq 0$ to be the probability that a molecule released from transmitter node $l$ hits receiver node $j$ in the next $k$-th time instance.  Since the signal received by receiver node $j$ at time $i$ is due to the transmissions from all transmitters by the time slot $i$, using similar steps as in the point to point case, the distribution of receiver node $j$ at time $i$ is
\begin{align}Y_{ji}\sim\Po(\lambda_0+\sum_{l=1}^s\sum_{u=0}^iX_{l,i-u}p_{l,j,u}),\qquad\forall i\geq 0\label{eqn1nnMM}\end{align}
where we are using the fact that $\Po(X_{l,i-u}p_{l,j,u})$ is the contribution from the transmitter $l$ at time $i-u$ that has reached receiver $j$ with a time delay of $u$ time slots. Equation \eqref{eqn1nnMM} is again a convolution and the triple $p_{l,j,u}$ specifies the impulse response of the LTI network given in Fig. \ref{Figure1.5}.

\begin{figure*}
\begin{center}
\includegraphics[scale=0.4,angle=0]{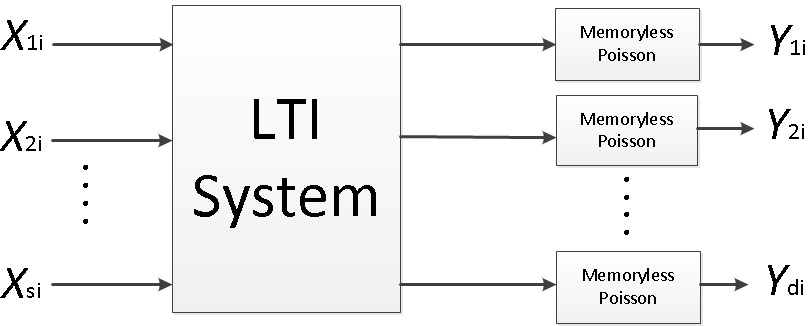}
\end{center}
\caption{An LTI-Poisson Model for a single-hop diffusion network}
\label{Figure1.5}
\end{figure*}

\emph{Memory-limited networks:} Assume that for some large enough $k$, molecules injected into the environment at times before $i-k-1$ are completely diffused in the environment and so their effect on the current outputs of the channel are negligible. This implies that the hitting probability $p_{l,j,u}=0$ for any $u>k$. In this case, the outputs at time $i$ depend only on the past $k$ inputs. More generally, we define a class of memory limited networks as follows:

\begin{definition}
Consider a single-hop network with inputs $\textbf{x}_{1:n}$ at time instances $1$ to $n$. We call this network a \emph{memory limited network} (MLN) of order $k$ if the output random variables at times $k+1$ to $n$ satisfy the following:
\begin{align}p(\textbf{y}_{k+1:n}|\textbf{x}_{1:n})=\prod_{i=k+1}^np(\textbf{y}_i|\textbf{x}_{i:i-k}).\label{eqn:FFSN}
\end{align}
\end{definition}

 This definition was used by  Verd\'{u}  in \cite{Verdu}  to model a linear ISI channel for the special case of multi-access channels. Indeed, the restriction given in Eq. \eqref{eqn:FFSN} is broad enough to include any memory-limited ISI finite state channel (linear or non-linear) when the state at time $i$ is determined by the inputs of its $k$ previous time slots.

We prove several of our theorems for MLN networks, which includes LTI-Poisson networks.


\section{Capacity of MLN networks}
\label{capacity}

Let $\mathcal{C}$ be the capacity region, including the set of all asymptotically achievable rates $(R_{ij})$ between all transmitter-receiver node pairs $(i,j)$. Our definition of the capacity region is the standard one with vanishing probability of error as the blocklength goes to infinity.  More specifically, an $(n,\epsilon)$ code for a network with $s$ transmitters and $d$ receivers consists of a set of messages $M_{ij}$ of length $n(R_{ij}-\epsilon)$ from transmitter $i$ to receiver $j$ ($1\leq i\leq s$, $1\leq j\leq d$), encoders $\mathcal{E}_i$ (for $1\leq i\leq s$) mapping messages $(M_{ij})_{1\leq j\leq d}$ to input $\textbf{X}_i=(X_{i1}, X_{i2}, \cdots, X_{in})$, and decoders $\mathcal{D}_j$ (for $1\leq j\leq d$) mapping the output $\textbf{Y}_j=(Y_{j1}, Y_{j2}, \cdots, Y_{jn})$ to messages $(\hat{M}_{ij})_{1\leq i\leq s}$, such that the probability that $\hat{M}_{ij}\neq M_{ij}$ for some $(i,j)$ is less than or equal to $\epsilon$.

In this section, we show that a computable characterization of the capacity MLN networks can be found in terms of  the corresponding memoryless counterparts, e.g. the capacity of a broadcast MLN can be expressed in terms of the capacity of a memoryless broadcast channel.

We begin with the following definition of a \emph{block memoryless} version of a channel with memory:

\begin{definition}\label{def1absdfr} Given an MLN channel defined by equation \eqref{eqn:FFSN} and a natural number $r$, consider 
 a block \underline{memoryless} channel with network description \begin{align}p(\textbf{y}_{k+1:k+r}|\textbf{x}_{1:k+r})=\prod_{i=k+1}^{k+r}
p(\textbf{y}_{i}|\textbf{x}_{i:i-k}),\label{eqn:newnewthm1}\end{align}
where $p(\textbf{y}_{i}|\textbf{x}_{i:i-k})$ on the right hand side is the description of the original network. In other words, given an MLN with $s$ transmitters and $d$ receivers, we create a virtual single-hop block memoryless network with the same transmitter and receivers. In each single use of this network, the $j$-th transmitter chooses $k+r$ symbols $x_{j, 1:k+r}$; in other words, the  $k+r$ symbols $x_{j, 1:k+r}$ combined together form one channel input. Once all the transmitters have commited their inputs, collectively shown by $\textbf{x}_{i:i-k}$, receiver $j$ gets $r$ output symbols $y_{j, k+1:k+r}$, which are collectively shown by $\textbf{y}_{k+1:k+r}$. 
\end{definition}

Before stating a computable characterization of the capacity region, given a region $\mathcal{R}$ and a real number $c$ we define $c\mathcal{R}$ to be the pointwise multiplication of vectors in $\mathcal{R}$ by the scalar $c$.

Then, we have the following theorem:

\begin{theorem}\label{thm1}
For any arbitrary $r\in\mathbb{N}$, the capacity $C$ of a MLN satisfies
\begin{align}\frac{r}{k+r}\mathcal{C}_r\subseteq \mathcal{C}\subseteq \mathcal{C}_r,\label{eqn:thm1}\end{align}
where
$\mathcal{C}_r$ is $1/r$ times the capacity region of the block memoryless  system of size $r$ as described in Definiton \ref{def1absdfr}.
\end{theorem}

See Appendix \ref{AppendixProofThm1} for a proof.

This is a complete and \emph{computable} characterization of the capacity region of an MLN channel in terms of  a corresponding memoryless channel. Given a certain accuracy level, we can find a suitable $r$ such that the lower and upper bound become close to each other within the given accuracy level.\footnote{As Cheng and  Verd\'{u}  note it is ``a not uncommon misconception is to dismiss limiting expressions for capacity as uncomputable" \cite{ChengVerdu} as there are examples of computable finite letter characterizations of capacity (including the one we propose in this paper). }

\begin{corollary}\label{cor1}
Consider the special case of a point to point channel. In this case,
\begin{align}\frac{r}{k+r}C_r\leq C\leq C_r,\label{eqn:thm1cor}\end{align}
where 
\begin{align}C_r=\frac{1}{r}\max_{p(x_{1:k+1})}I(X_{1:k+r}; Y_{k+1:k+r}).\label{eqn:r1}\end{align}
\end{corollary}

\begin{remark}
The proof of the upper and lower bounds follow from more specialized versions for deterministic LTI channels \cite{Massey, Goldsmith} or specific networks (e.g. MAC MLN in \cite{Verdu}); that is, we set or reset the channel using $k$ consecutive inputs.  Other related work which considers networks with more general models of memory than LTI \cite{Gallager, Verdu, Dabora,  Chen} take on a different approach.  That is, the outer bound derivation begins with an $n$-letter expression which often resembles the capacity region of the {\it memoryless} case.  Specific computations, typically exploiting entropy bounds are then developed.  In contrast, our outer bound is completely operational.  Note that one need not know the explicit form of the capacity in order to prove or employ our bounds.
\end{remark}

Computing $C_r$ in Eq. \eqref{eqn:r1} is difficult to compute for large values of $r$, particularly when $X$ is a continuous random variable. However, we shall be able to derive upper and lower bounds on $C_r$ which will yield meaningful bounds on capacity for our LTI Poisson channel. This approach is taken to derive the following theorem: 

\begin{theorem}\label{thm2} The capacity of a point to point memory limited channel of order $k$ is bounded as  follows
\begin{align}
&\max_{p(x_1, \cdots, x_{k+1})\in\mathcal{P}} ‎I(X_{k+1};Y_{k+1}|X_{1:k})\leq C\leq  \max_{p(x_1, \cdots, x_{k+1})\in\mathcal{P}} ‎I(X_{1:k+1};Y_{k+1}), \label{prob1}
\end{align}
where $\mathcal{P}$ be the set of joint pmfs $p(x_1, \cdots, x_{k+1})$ satisfying $$p_{X_1\cdots X_k}(x_1, \cdots, x_k)=p_{X_2\cdots X_{k+1}}(x_1, \cdots, x_k),$$ for every values of $x_1, \cdots, x_k$.
\end{theorem}
The proof can be found in Appendix \ref{AppendixProofThm2} . 
\begin{remark}
The form of the upper bound part of Eq. \eqref{prob1} is similar to that of Eqs. \eqref{eqn:thm1} and \eqref{eqn:r1} in Theorem \ref{thm1} for the case of $r=1$ (we denote this capacity as $C_1$). However  the upper bound part of Eq. \eqref{prob1} has its maximum over $p(x_1, \cdots, x_{k+1})\in\mathcal{P}$ whereas in $C_1$, the maximum is over all joint distributions. Therefore  the upper bound in Eq. \eqref{prob1} is less than or equal to $C_1$. 
\end{remark}

‎\section{Symmetrized KL Divergence Upper Bound}‎
\label{symm}

Expressing the capacity in terms of the maximum of the mutual information over a certain class of probability distributions is inadequate from a practical perspective in many networks and channels of interest. This inadequacy extends to some memoryless channels such as memoryless Poisson channels \cite{Lapidoth2011, LapidothMoser2009}. However, this need is more pronounced in the computationally burdensome problems of channels with state 
as previously discussed in Section \ref{capacity}. To address this  problem, in this section, we propose a new upper bound on mutual information based on the  symmetrized Kullback-Leibler (KL) divergence. 

 The basic idea of the upper bound is simple: note that $I(X;Y)=D(p(x,y)\|p(x)p(y))$, where $D(p\|q)$ is the KL divergence and is defined via the equation $D(p\|q)=\sum_{i}p_i\log p_i/q_i.$ Thus, if we can find an upper bound on KL divergence functional $D(\cdot \| \cdot)$, we can find an upper bound on mutual information by evaluating the upper bound at the pair $\big(p(x,y), p(x)p(y)\big)$. One could then seek  known divergences that serve as an upper bound on the KL divergence; two natural choices are the \emph{$\alpha$-Renyi divergence} for $\alpha>1$, and the \emph{symmetrized divergence}. The former choice $D_{\alpha}(p(x,y)\|p(x)p(y))$, in a slightly modified form, is called the  $\alpha$-Renyi mutual information and does show up in the context of strong converses and error exponents (e.g. see \cite{Beigi, Gupta}). However, given a channel $p(y|x)$, computing the maximum of $$\max_{p(x)} D_{\alpha}(p(x,y)\|p(x)p(y)),$$ for $\alpha>1$ does not seem to be any easier or more insightful than performing the optimization for $\alpha=1$ (which yields Shannon capacity). On the other hand, the symmetrized KL divergence will offer some computational advantages.

\begin{definition}[Symmetrized KL divergence] Let
$$D_{\mathsf{sym}}(p\|q):=D(p\|q)+D(q\|p).$$
\end{definition}
Clearly $D_{\mathsf{sym}}(p\|q)\geq D(p\|q)$. 

Let us define the following upper bound on the mutual information:
$$A(p(x,y))=D_{\mathsf{sym}}(p(x,y)\|p(x)p(y))\geq D\big(p(x,y)\|p(x)p(y)\big)=I(X;Y).$$

Similarly, for a channel $p(y|x)$ we define
$$A(p(y|x))=\max_{p(x)}D_{\mathsf{sym}}(p(x,y)\|p(x)p(y))\geq \max_{p(x)}I(X;Y)=C(p(y|x)).$$

The quantity $A(p(y|x))$ is always an upper bound on the capacity. It is straightforward to show that
\begin{align*}
D_{\mathsf{sym}}\big(p(x,y)\|p(x)p(y)\big)&=\sum_{x,y}p(x,y)\log\frac{p(x,y)}{p(x)p(y)}+\sum_{x,y}p(x)p(y)\log\frac{p(x)p(y)}{p(x,y)}
\\&=\sum_{x,y}p(x,y)\log p(y|x) - \sum_{x,y}p(x,y)\log p(y)
\\&\qquad+\sum_{x,y}p(x)p(y)\log p(y)+\sum_{x,y}p(x)p(y)\log\frac{1}{p(y|x)}
\\&=\sum_{x,y}p(x,y)\log p(y|x) - \sum_{y}p(y)\log p(y)
\\&\qquad+\sum_{y}p(y)\log p(y)+\sum_{x,y}p(x)p(y)\log\frac{1}{p(y|x)}
\\&=\sum_{x,y}\big[p(x,y)-p(x)p(y)\big]\log(p(y|x)).
\\&=\mathbb{E}_{p(x,y)}\log p(Y|X) - \mathbb{E}_{p(x)p(y)}\log p(Y|X).
\end{align*}
 The noteworthy feature here is that $D_{\mathsf{sym}}\big(p(x,y)\|p(x)p(y)\big)$ is a quadratic expression in $p(x)$ whereas $D\big(p(x,y)\|p(x)p(y)\big)$ is a logarithmic curve. To see that $D_{\mathsf{sym}}\big(p(x,y)\|p(x)p(y)\big)$ is a quadratic expression, observe that given a fixed channel $p(y|x)$ we have
\begin{align*}\sum_{x,y}\big[p(x,y)-p(x)p(y)\big]\log(p(y|x))&=\sum_{x,y}\big[p(x,y)-p(x)\sum_{\tilde{x}}p(\tilde{x},y)\big]\log(p(y|x))
\\&=\sum_{x,y}
p(x)p(y|x)\log(p(y|x))\\&\qquad-
\sum_{x,\tilde{x},y}p(x)p(\tilde{x})p(y|\tilde{x})\log(p(y|x)).
\end{align*}

 Therefore $D_{\mathsf{sym}}\big(p(x,y)\|p(x)p(y)\big)$ over $p(x)$ should be computationally easier.  For this reason, it serves as a convenient upper bound for complicated expressions of mutual information, like the finite-letter characterization given in the previous section.

\begin{remark}\label{rmk:rmkrmk}
$A(p(y|x))=C(p(y|x))$  if and only if the capacity of the channel is zero. This is because $D_{\mathsf{sym}}(p(x,y)\|p(x)p(y))= D\big(p(x,y)\|p(x)p(y)\big)$ implies $D\big(p(x)p(y)\|p(x,y)\big)=0$, in turn implying that $D\big(p(x,y)\|p(x)p(y)\big)=0$. Therefore this upper bound on capacity has potential for the low SNR regime. The low SNR regime is of relevance in molecular communication setups. For Poisson and Guassian channels that we numerically simulated, in the low SNR regime, the extra term $D\big(p(x)p(y)\|p(x,y)\big)$ is smaller or of the same order as $D\big(p(x,y)\|p(x)p(y)\big)$. Thus, roughly speaking if for instance the capacity is about 0.001, the upper bound will be less than or equal to 0.002 or 0.003, which is of the same order as 0.001.  Numerical results verify this ratio of the upper bound and the capacity.

\end{remark}

To demonstrate the benefit of the proposed upper bound, we consider two well-studied channels: the Gaussian channel and the memoryless Poisson channel.

\begin{example}\label{example1}
Consider a point to point Gaussian channel $p(y|x)$ where $Y=X+N$ for some Gaussian noise $N\sim \mathcal{N}(\mu, \sigma^2)$, and any input pmf $p(x)$. Then $A(p(x,y))=\mathsf{Var}(X)/\sigma^2$. The proof of this derivation is given in Appendix \ref{AppendixProofExample1}. Hence, if we have a power input constraint $P$, we get that
$$C\leq \max_{p(x)}A(p(x,y))=SNR,$$
where $SNR=P/\sigma^2$. This upper bound is within a factor two of the capacity in the low SNR regime, and hence is of the same order: as the capacity of Gaussian channel is $ \log(1+P/\sigma^2)/2$, in the low SNR regime we have $\log(1+P/\sigma^2)/2\simeq P/\sigma^2/2=SNR/2$.
\end{example}

\begin{example}\label{example2}
Consider a point to point Poisson channel $p(y|x)$ where $Y=Poisson(X+\lambda_0)$, and any input pmf $p(x)$. Then, the symmetrized KL divergence upper bound has the following compact formula:
\begin{equation}
I(X;Y)\leq A(p(x,y))=\mathsf{Cov}(X+\lambda_0, \log(X+\lambda_0)),
\end{equation}
where $\mathsf{Cov}(X,Y)=\mathbb{E}[XY]-\mathbb{E}[X]\mathbb{E}[Y]$. Further, for a Poisson channel with average intensity constraint $\alpha$ and maximum intensity constraint $\mathsf{A}$ we have
$$A_{\mathsf{Poisson}}(p(y|x)):=\max_{\substack{p(x):\\ E[X]=\alpha,~~ 0 \leq X \leq \mathsf{A}}}A(p(x,y))=\begin{cases}\frac{\alpha}{\mathsf{A}}(\mathsf{A}-\alpha)\log(\frac{\mathsf{A}}{\lambda_0}+1), & \alpha< \mathsf{A}/2
\\
\frac{\mathsf{A}}{4}\log(\frac{\mathsf{A}}{\lambda_0}+1), & \alpha\geq \mathsf{A}/2
\end{cases}$$
Hence,
$$C=\max_{\substack{p(x):\\ E[X]=\alpha,~~ 0 \leq X \leq \mathsf{A}}}I(X;Y)\leq A_{\mathsf{Poisson}}(p(y|x)).$$
The derivation of the expressions given in this example is given in Appendix \ref{AppendixProofExample2}.
\end{example}

To the best of our knowledge, bounding the capacity of the Poisson channel in the low SNR regime with finite $\mathsf{A}$ has not been previously considered.  In  \cite{Lapidoth2011},  an upper bound for low SNR Poisson channel for $\mathsf{A}=\infty$  is provided; however, this expression is very complex. One can also use the above upper bound to prove the upper bound of Eq. (6) in \cite{Lapidoth2011}.

A plot for the capacity and the upper bound is given in Figure \ref{UpperCapacityComparison}. It is observed the gap between the bound and the capacity decreases to about $0.4$ in the logarithmic scale as noise parameter increases; hence the upper bound is about 3 times the capacity, meaning that they are of the same order. Also, the figure demonstrates that for a fixed $\lambda_0$, the gap decreases for smaller values of $\alpha$. Note that both increasing $\lambda_0$ and decreasing $\alpha$ can be interpreted as decreasing SNR.

\begin{figure*} [t]
\begin{center}
\includegraphics[scale=0.32,angle=0]{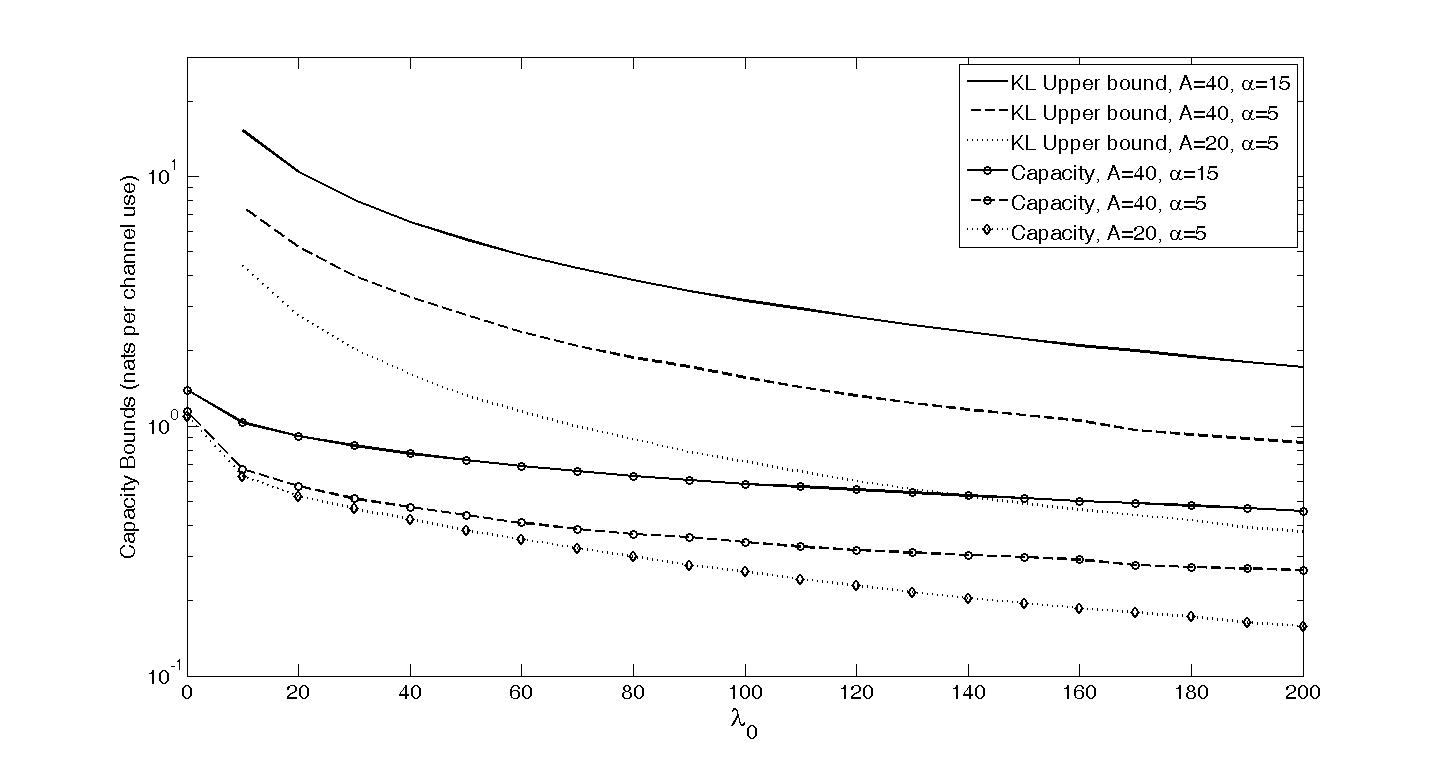}
\end{center}
\caption{Capacity and Symmetrized divergence upper bound in terms of $\lambda_0$ for memoryless Poisson channel with $A=40$ and $\alpha=5, 15$.}
\label{UpperCapacityComparison}
\end{figure*}

\begin{remark}
An interesting observation is that the pmf that is maximizing $A(p(x,y))$ in Example \ref{example2}  is 
 $p(x)=\frac{\alpha}{A}\delta(x-A)+(1-\frac{\alpha}{A}) \delta(x)$. This pmf is exactly the same pmf that maximizes $I(X;Y)$ of a Poisson channel at low SNRs, when $\lambda_0$ is large and $\alpha \leq \frac{A}{2}$ \cite{Cao}. Therefore mutual information and its upper bound reach their maximums at exactly the same point! The same phenomenon occurs for a BSC(p) channel where 
\begin{align*}
A_{\mathsf{BSC}}=-\log_2(\sqrt{p(1-p)})-h(p).
\end{align*}
and the maximum occurs at uniform input pmf distribution. 
\end{remark}

\begin{remark}\label{remarkTopsoe} A widely used technique for proving an upper bound uses the following relation \cite{Topsoe, AllertonUpper}: $$D(p(x,y)\|p(x)p(y))=\min_{r(y)} D(p(x,y)\|p(x)r(y)).$$
Therefore any arbitrary $r(y)$ gives us an upper bound. This idea can be potentially mixed with ours: for any arbitrary $r(y)$ and any upper bound on the divergence functional, we can evaluate the upper bound at the pair $\big(p(x,y), p(x)r(y)\big)$.
\end{remark}

\subsection{Some properties of the upper bound}
\label{section:someproperties}
Even though the proposed upper bound is not equal to the mutual information, there are similarities between the behaviours of both, when viewed as a function  of the channel for a fixed input distribution. Below we provide some of these properties, as well as an alternative proof of the upper bound via Jensen's inequality.

\begin{enumerate}
\item Similar to $I(X;Y)$, the upper bound $A(p(x,y))$ is convex in $p(y|x)$ for a fixed $p(x)$. Indeed both $D(p(x,y)\|p(x)p(y))$ and $D(p(x)p(y)\|p(x,y))$ are convex in $p(y|x)$ because the KL divergence is convex in its input pair.
\item Similar to capacity, $A$ factorizes for product channels, i.e. $$A\big(p(y_1|x_1)p(y_2|x_2)\big)=A\big(p(y_1|x_1)\big)+A\big(p(y_2|x_2)\big).$$ In other words, the maximum
$$\max_{p(x_1, x_2)}D_{\mathsf{sym}}(p(x_1,x_2,y_1, y_2)\|p(x_1,x_2)p(y_1,y_2)),$$
occurs at some product distribution $p(x_1, x_2)=p(x_1)p(x_2)$. The proof can be found in Appendix \ref{AppendixPropertiesProof}. 

\item The fact that $\sum_{x,y}\big[p(x,y)-p(x)p(y)\big]\log(p(y|x))$ is an upper bound on $I(X;Y)$ can also be seen from Jensen's inequality on the log function; this is not unexpected since many of the inequalities on divergence can be proved using Jensen's inequality. The proof can be found in Appendix \ref{AppendixPropertiesProof}. 
\end{enumerate}

\section{Capacity of LTI-Poisson channel}
\label{LTIpoisson}
Let us denote the capacity of the LTI-Poisson channel by $\mathcal{C}(\mathsf{A}, \alpha, \textbf{p}, \lambda_0)$. Assuming that $p_j=0$ for $j>k$, to estimate $\mathcal{C}(\mathsf{A}, \alpha, \textbf{p}, \lambda_0)$, from Corollary \ref{cor1} we should compute 
\begin{align*}C_r=\frac{1}{r}\max_{p(x_{1:k+r})}I(X_{1:k+r}; Y_{k+1:k+r}),\end{align*}
where the maximum should be taken over pmfs satisfying $$X_i\in [0, \mathsf{A}], \qquad \frac{1}{k+r}\sum_{i=1}^{k+r}\mathbb{E}[X_i]\leq \alpha,$$ and $Y_i=\mathsf{Poisson}(\lambda_0+\sum_{j=0}^kp_{j}X_{i-j})$. Alternatively, one can use the upper and lower bounds on $C_r$. 

\subsection{Numerical  results}

\begin{figure*}
\begin{center}
\includegraphics[scale=0.3,angle=0]{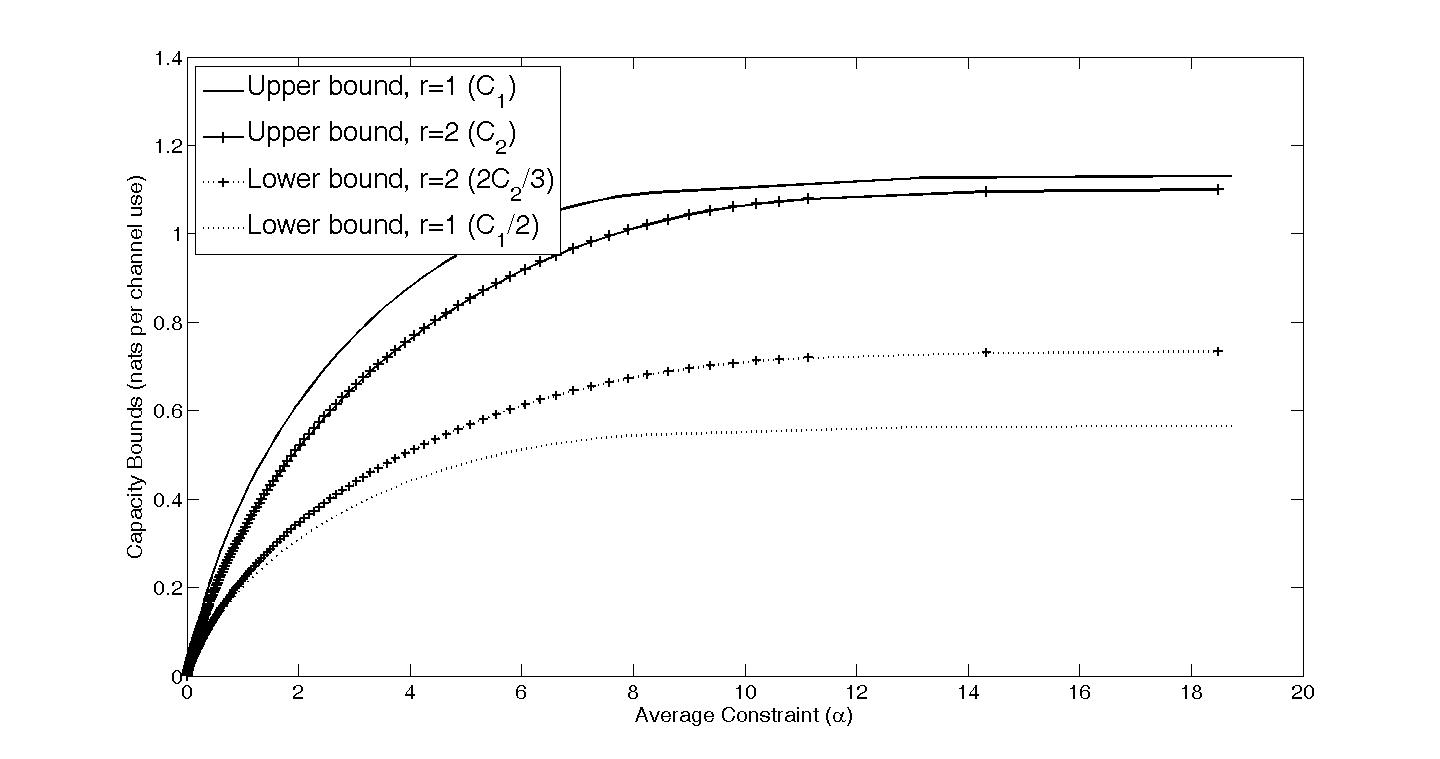}
\end{center}
\caption{Capacity lower and upper bounds (Theorem 1) in terms of average constraint $\alpha$ for LTI-Poisson channel with $\lambda_0=5$, $\mathsf{A}=40$ for $r=1, 2$.}
\label{Figure2a}
\end{figure*}

\begin{figure*}
\begin{center}
\includegraphics[scale=0.44,angle=0]{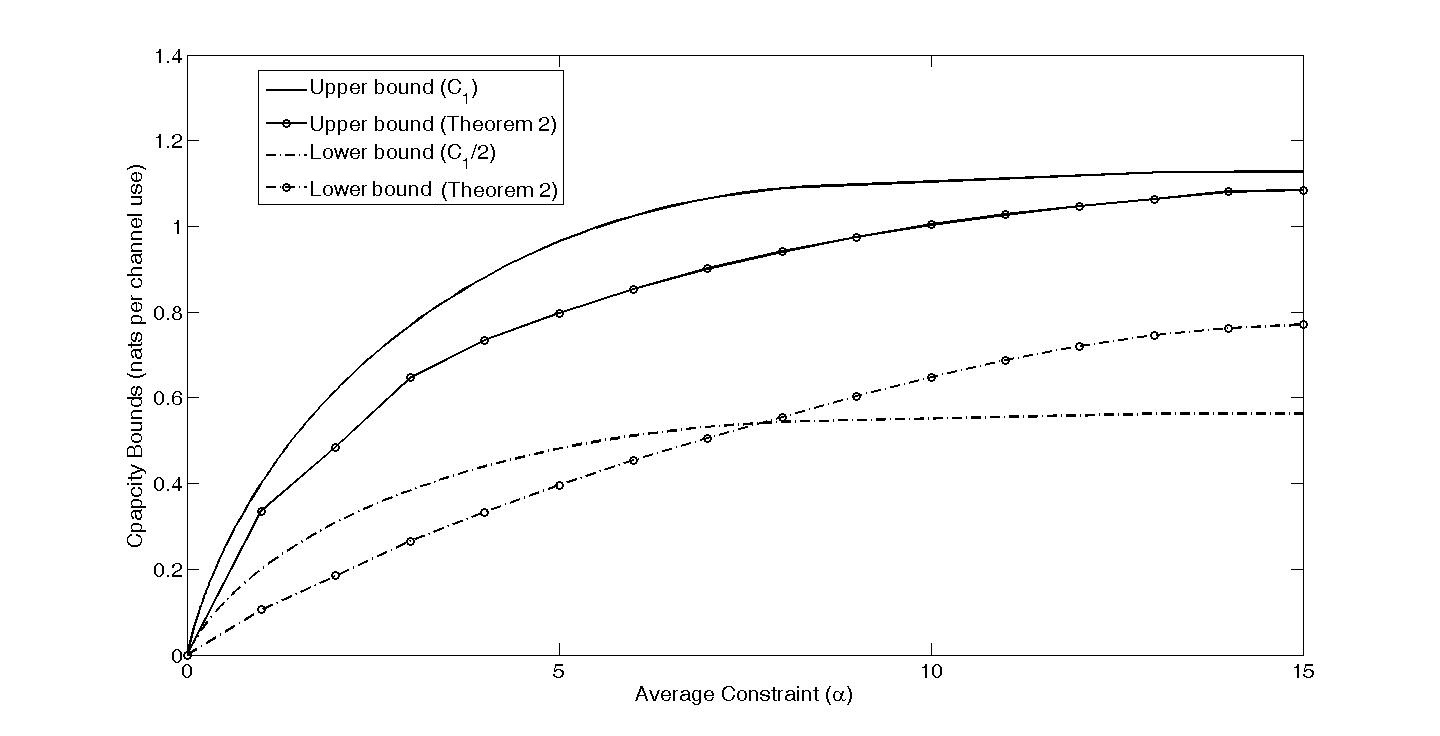}
\end{center}
\caption{Comparison of proposed upper and lower bounds in Theorem 2 with that of Theorem 1 (for $r=1$) in terms of $\alpha$ for LTI-Poisson channel with $\mathsf{A}=40$ and $\lambda_0=5$.}
\label{Figure3a}
\end{figure*}

\begin{figure*} 
\begin{center}
\includegraphics[scale=0.33,angle=0]{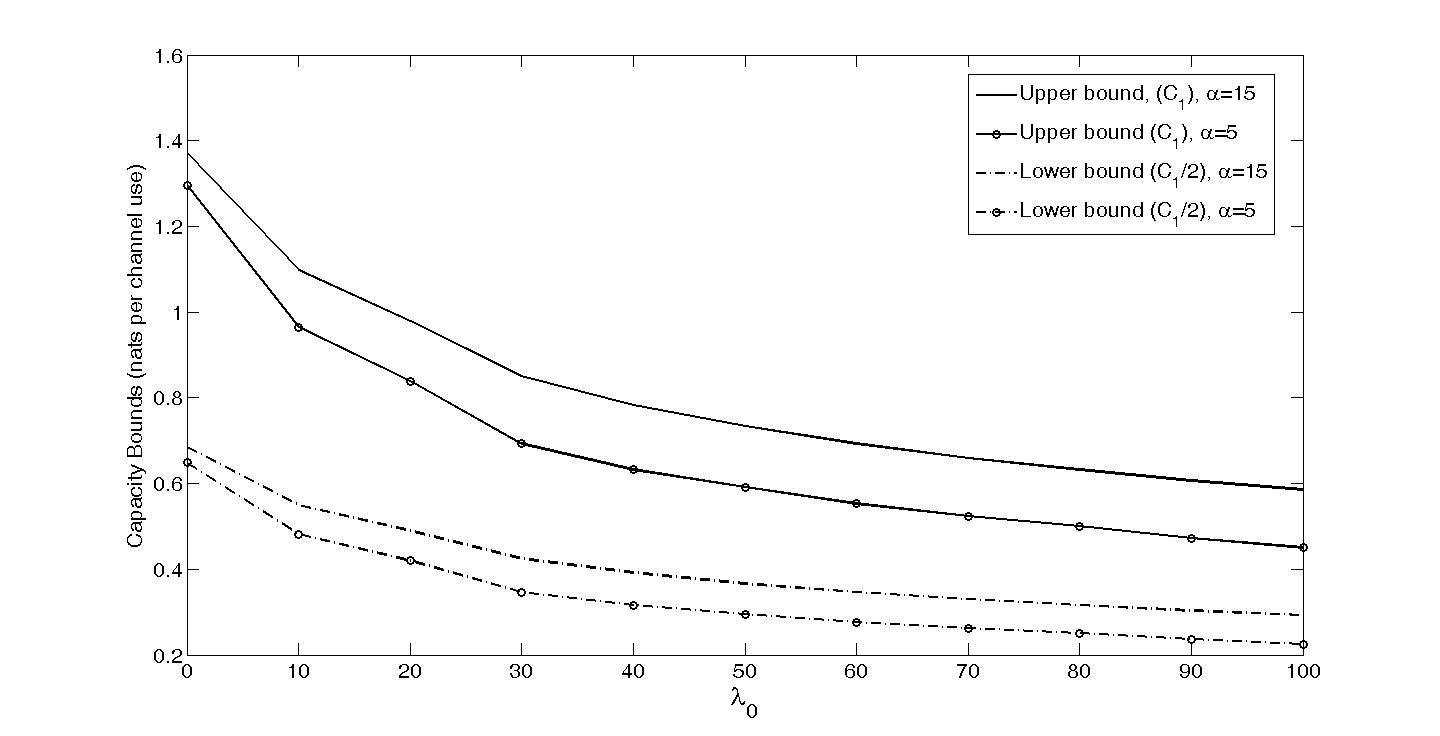}
\end{center}
\caption{Capacity upper and lower bounds proposed in Theorem 1, $C_r$ and $\frac{K}{r+K} C_r$ for $r=1$ in terms of $\lambda_0$ for LTI-Poisson channel with $\mathsf{A}=40$, $\alpha=1$ and $\alpha=5, 15$. }
\label{Figure4a}
\end{figure*}
 To evalute the proposed capacity bounds, we have assumed a transmitter to receiver distance of $8 \mu m$.  The medium diffusion constant is taken to be $2‎\times‎10^{-8} cm^2/s$, which is a practical value \cite{practical-diff}. The memory of the channel is assumed to be $k=1$, and the time slot of each channel use is $1.5 s$. To find the optimal input distributions maximizing $C_r$, the numerical Blahut-Arimoto algorithm (BA) \cite{blahut} is used.   Fig. \ref{Figure2a} depicts the lower and upper bounds proposed in Corollary \ref{cor1} in terms of average power constraint $\alpha$ for an LTI-Poisson channel with the parameters $\lambda_0=5$, $A=40$ for  $r=1, 2$. Observe that both upper and lower bounds are increasing as $\alpha$ increases. Observe that the bounds (and equivalently the capacity) saturate; this is expected since we know that the capacity does not increase when we increase $\alpha$ beyond $\mathsf{A}$. As expected, for $r=2$, the gap between the lower and upper bounds decreases in comparison with $r=1$ and the bounds approach the capacity. Note that the upper bound ($C_r$) is decreasing in terms of $r$ as the figure shows that $C_2$ is always higher than $C_1$. Similarly, the lower bound $\frac{r}{k+r}C_r$ is increasing for $r=1,2$.

Fig. \ref{Figure3a} compares the proposed upper bound and lower bound in Theorem \ref{thm2} with that of Corollary \ref{cor1} (for $r=1$) , \emph{i.e.} $C_1$ and $C_1/2$. As expected, the upper bound in Theorem \ref{thm2} is tighter than $C_1$ as the maximum is taken  over a smaller set of distributions. However, the lower bound improves the $C_1/2$ for larger values of $\alpha$.

Fig. \ref{Figure4a} demonstrates the behavior of the proposed capacity bounds in Theorem 1 in terms of noise parameter of $\lambda_0$. Observe that the gap between the upper bound and lower bound decreases as $\lambda_0$ increases. Also, the figure shows that the sensitivity of the bounds to the noise parameter $\lambda_0$ is higher for smaller noise mean values.

\subsection{Some analytical results}
By definition, the capacity $\mathcal{C}(\mathsf{A}, \alpha, \textbf{p}, \lambda_0)$ is increasing in $\mathsf{A}$ and $\alpha$. It is also decreasing in $\lambda_0$. This is intuitive and can be shown concretely using methods similar to those in Eqs. (51)-(57) of \cite{LapidothMoser2009}. To study the behavior of the capacity of LTI-Poisson channel in terms of $\textbf{p}$, first observe that 
from Eq. \eqref{eqn1nn}, for any $\beta> 0$
$$\mathcal{C}(\mathsf{A}, \alpha, \textbf{p}, \lambda_0)=\mathcal{C}(\beta A, \beta \alpha,\frac{1}{\beta}\textbf{p}, \lambda_0).$$
Therefore when studying $\mathcal{C}(\mathsf{A}, \alpha, \textbf{p}, \lambda_0)$, without loss of generality we can assume that $\sum_{i}p_i=1$.

 The  following lemma provides a characterization for $C_1$ for the LTI Poisson model. Since $C_1$ shows up in both the lower and upper bounds of Theorem 1, the lemma allows us to easily compute lower and upper bounds on the capacity region.
\begin{lemma}\label{lemma1}
To compute $C_1$ for a given average power constraint $\alpha$, but unlimited peak power constraint $\mathsf{A}=\infty$, instead of maximizing over all pmfs $p(x_{1:k+1})$, it suffices to look at random variables of the form $X_i=\beta_i X$ for some rv $X$ and non-negative reals $\beta_i$. 
\end{lemma}
\begin{proof}
We have
\begin{align*}C_1=\max_{p(x_{1:k+1})}I(X_{1:k+r}; Y_{k+1}),\end{align*}
where 
\begin{align*}Y_{k+1}\sim\Po(\lambda_0+\sum_{j=0}^{k+1}p_{j}X_{i-j}).\end{align*}
Take an arbitrary pmf $p(x_{1:k+1})$. Let $S=\sum_{j=0}^{k+1}p_{j}X_{i-j}$. Then, we have the Markov chain $X_{1:k+1}\rightarrow S\rightarrow Y_{k+1}$. Hence, $I(X_{1:k+r}; Y_{k+1})\leq I(S; Y_{k+1})$. On the other hand, since $S$ is a function of $X_{1:k+r}$, we have $I(X_{1:k+r}; Y_{k+1})\geq I(S; Y_{k+1})$. Therefore 
$$I(X_{1:k+r}; Y_{k+1})=I(S; Y_{k+1}).$$
Let 
$\tilde{X}_i=\beta_i\cdot S$ where $\beta_i=\mathbb{E}[X_i]/\mathbb{E}[S]$ for $i=1,2, \cdots, k+1$. Further let
$$\tilde{Y}_{k+1}\sim\Po(\lambda_0+\sum_{j=0}^{k+1}p_{j}\tilde{X}_{i-j}).$$
Clearly $\tilde{X}_i$ are proportional to each other.  These variables are a legitimate choice as input pmf since $\mathbb{E}[\tilde{X}_i]=\mathbb{E}[X_i]$ and hence the average power constraint is preserved. Further
\begin{align*}\tilde{S}&=\sum_{j=0}^{k+1}p_{j}\tilde{X}_{i-j}
\\&=S\cdot \sum_{j=0}^{k+1}p_{j}\mathbb{E}[X_{i-j}]/E[S]
\\&=S.
\end{align*}
and hence $I(X_{1:k+r}; Y_{k+1})=I(\tilde{X}_{1:k+r}; \tilde{Y}_{k+1})$. This will complete the proof. 
\end{proof}

Observe that the channel is the cascade of an LTI filter defined by $\textbf{p}$ with a memoryless Poisson channel; therefore one may guess that its capacity is less than or equal to the capacity of a memoryless Poisson (which corresponds to the special case of $p_0=1, p_i=0,~i\geq 1$). But this conjecture  requires a proof as we are dealing with channels with input constraints.\footnote{If we have  $X-X'-Y$, the capacity of the channel $p(y|x)$ is not necessarily less than the capacity of $p(y|x')$ when we impose input constraints on $X$ and $X'$.}

\begin{theorem}\label{prop2}

Assuming that $\sum_{i}p_i=\sum_{i}p'_i=1$, we have 
$$\mathcal{C}(\mathsf{A}, \alpha, \textbf{p}, \lambda_0)\geq \mathcal{C}(\mathsf{A}, \alpha, \textbf{p}', \lambda_0),$$
if $\textbf{p}'= \textbf{p}\star‎ \textbf{q}$ for some non-negative sequence $\textbf{q}=(q_0, q_1, \cdots)$.
\end{theorem}

The proof is given in Appendix
 \ref{AppendixProofClaim2}.
 The above theorem allows us to define a notion of ``degradedness" for diffusion channels \emph{with power constraints}; here $\textbf{p}'$ is a degraded version of $\textbf{p}$. The capacity of a channel forms an upper bound on the capacity of its degraded versions. 

\section{Conclusions}
In this paper, we provided several capacity results for  diffusion based molecular communications networks under the LTI-Poisson model.  We first provided a computable characterization of the capacity region for a class of memory limited networks that generalizes both the linear ISI channel and the LTI-Poisson model. As expected, and confirmed by our numerical results, the gap between the lower and upper bounds decreases as $\lambda_0$ increases and also as $r$ increases and the bounds approach the capacity.  Next, we considered the special case of a point to point channel and developed some lower and upper bounds on the capacity for this case. Finally, we derived a new upper bound on mutual information using symmetrized Kullback-Leibler divergence (KL divergence). The proposed upper bound, similar to $I(X;Y)$, is convex in $p(y|x)$ for a fixed $p(x)$ and similar to capacity, factorizes for product channels. Our numerical results indicate that this upper bound works well for channels with small capacity (which can occur in MC systems). The technique we use is fairly general; a case for it is made by finding an upper bound for the Poisson channel with large background noise.

\bibliographystyle{ieeetr}
\bibliography{reftest}

\appendices
\section{Proof of Theorem \ref{thm1}}
\label{AppendixProofThm1}
 The proof of the lower bound is a direct generalization of the one given in \cite{Verdu}. However for the sake of completeness, we write it here.

To prove the lower bound, we create a memoryless channel whose capacity is a subset of the original channel. For the upper bound, we create another  memoryless channel whose capacity  subsumes that of the original channel. The lower bound is based on the idea of ``channel depreciation via deletion" and the upper bound is based on  the idea of ``channel enhancement via insertion". The idea of ``channel depreciation via deletion"  is known  and used in \cite{Verdu}. The idea of ``channel enhancement via insertion" is similar in spirit, but we have not explicitly seen it before.

\textbf{Lower bound:} Take some $r\in\mathbb{N}$ and let us partition time into frames or blocks of size $k+r$, i.e. one block for time instances 1 to $k+r$, one block for time instances $k+r+1$ to $2(k+r)$, etc. We depreciate the channel by deleting the output $\mathbf{Y}_i$'s for the first $k$ time instances of each block; i.e. the new channel has inputs $\mathbf{X}_1, \mathbf{X}_2, \cdots$ but outputs $\mathbf{Y}_{k+1}, \mathbf{Y}_{k+2}, \cdots, \mathbf{Y}_{k+r}$ and then $\mathbf{Y}_{2k+r+1}, \mathbf{Y}_{2k+r+1}, \cdots, \mathbf{Y}_{2k+2r}$, etc. 

Clearly the capacity of the new channel is less than or equal to the capacity of the original channel. 

Next note that the outputs in each block depends only on inputs in the same block because $\mathbf{Y}_i$ is produced from $\mathbf{X}_i, \mathbf{X}_{i-1}, \mathbf{X}_{i-2}, \cdots, \mathbf{X}_{i-k}$.

 In other words, $\mathbf{Y}_{2k+r+1}, \mathbf{Y}_{2k+r+1}, \cdots, \mathbf{Y}_{2k+2r}$ depends only on $\mathbf{X}_{k+r+1}, \mathbf{X}_{k+r+2}, \mathbf{X}_{2k+2r}$ and not inputs from the other blocks. Therefore the new channel is ``memoryless" (in terms of blocks) and its capacity is known. Its capacity region is equal to $r\mathcal{C}_r$ where $\mathcal{C}_r$ is defined in the statement of the theorem.
Each block corresponds to $k+r$ uses of the original channel, therefore the capacity of the original channel is greater than or equal to 
$$\frac{r}{k+r}\mathcal{C}_r.$$

\textbf{Upper bound:} Take some $r\in\mathbb{N}$ and let us partition the time into frames or blocks of size $r$, i.e. one block for time instances 1 to $r$, one block for time instances $r+1$ to $2r$, etc. 
We enhance the channel by introducing $k$ fictitious inputs at the beginning of each block (which we call ``state-reset" inputs). In other words, we create a new channel as follows:  in the beginning of the first block we introduce  fictitious inputs $\mathbf{X}'_1(1), \mathbf{X}'_2(1), \cdots, \mathbf{X}'_k(1)$  where for instance $\mathbf{X}'_i(1)$ is a vector of size $s$, formed by the $i$-th fictitious input of all of the $s$ transmitters in the first block, etc.; in the beginning of the second block, we introduce  fictitious inputs  $\mathbf{X}'_1(2), \mathbf{X}'_2(2), \cdots, \mathbf{X}'_k(2)$, etc. Next, we also change the way the outputs $\mathbf{Y}_i$ are produced. At the beginning of each block, we assume that the state is suddenly set to $(\mathbf{X}'_1, \mathbf{X}'_2, \cdots, \mathbf{X}'_k)$ of the block, meaning that the network at the beginning of each block changes its behaviour, pretending that the $k$ fictitious inputs have been the actual last $k$ inputs of the previous block. The system continues to operate using the subsequent inputs and this initialization of the state. 

To sum this up, we are allowing the transmitters to choose the state sequence at the beginning of each block. The new channel is an enhancement of the original channel, since the transmitters can simply choose to choose the fictitious inputs to be the ones that have occured in the end of the previous block, i.e. in the $j$-th block:
$$(\mathbf{X}'_1(j), \mathbf{X}'_2(j), \cdots, \mathbf{X}'_k(j))\leftarrow (\mathbf{X}_{jr-k+1}, \mathbf{X}_{jr-k+2}, \cdots, \mathbf{X}_{jr}),$$
where $\mathbf{X}_{jr}$ is the set of inputs at time $jr$ (since each block is of size $r$, the last input vector of the $j$-th block would happen at time $jr$).

We call this \emph{enhancement by insertion} since we are inserting the new inputs $\mathbf{X}'_i$~s. Observe that in the new channel, blocks do not affect each other. The $\mathbf{Y}_i$ in each block depends only on $\mathbf{X}'_i$ and $\mathbf{X}_i$ of the same block; hence the new channel is memoryless over the blocks. Therefore, the capacity of the original channel is bounded from above by the capacity of the new channel. Each block corresponds to $r$ uses of the original channel, therefore the capacity of the original channel is less than or equal to $\mathcal{C}_r$.

\section{Proof of Theorem \ref{thm2}}
\label{AppendixProofThm2}

\textbf{Proof of the lower bound:} To prove the lower bound, it suffices to show that the given rate is less than or equal to $C_r$ for all $r$. This is because the limit of $(r/(m+r)) C_r$ as $r$ goes to infinity, is the same as the limit of $C_r$ as $r$ goes to infinity.

Take some arbitrary $q_{X_{1:k+1}}(x_{1:k+1})\in\mathcal{P}$, and let us choose the following joint pmf on $X_1, X_2, \cdots, X_{k+r}$:
$$p(x_{1:k+r}):=q(x_{1:k+1})\prod_{i=k+2}^{k+r}q_{X_{k+1}|X_{k}X_{k-1}\cdots X_{1}}(x_i|x_{i-1:i-k}).$$ By using induction on $i$, one can use the property of $q(x_{1:k+1})$ being in $\mathcal{P}$ to show that for any arbitrary $k+1\leq i\leq k+r$, we have
\begin{align}p(x_{i-k:i})=q_{X_{1:k+1}}(x_{i-k:i}).\label{eqn:anewa4}\end{align}
Due to how this pmf is defined, we have the Markov chain 
\begin{align}
X_i \rightarrow X_{i-1:i-k}\rightarrow X_{i-k-1}, X_{i-k-2} \cdots\label{eqn:anewa1}\end{align}
Next, for this joint pmf we would like to compute
\begin{align*}I(X_{1:k+r}; Y_{k+1:k+r}).\end{align*}
This would be a lower bound on $C_r$. The key to bounding this expression from below is the following observation: since the channel is memory limited we have $p(y_i|x_i, x_{i-1}, \cdots)=p(y_i|x_{i:i-k})$ and hence
\begin{align}0&=I(Y_i;X_{i-k-1}X_{i-k-2}\cdots |X_{i:i-k})\nonumber
\\&= I(X_iY_i;X_{i-k-1}X_{i-k-2}\cdots |X_{i-k:i-1}) - I(X_i;X_{i-k-1}X_{i-k-2}\cdots |X_{i-k:i-1})\label{eq:chainrule2}
\\&= I(X_iY_i;X_{i-k-1}X_{i-k-2}\cdots |X_{i-k:i-1}) \label{anewa2}
\\&\geq I(Y_i;X_{i-k-1}X_{i-k-2}\cdots |X_{i-k:i-1}), \label{eqneqnanewa2}
\end{align}
where Eq. \eqref{eq:chainrule2} follows from the chain rule, Eq. \eqref{anewa2} follows from Eq. \eqref{eqn:anewa1}. Therefore from Eq. \eqref{eqneqnanewa2}
$$I(Y_i;X_{i-k-1}X_{i-k-2}\cdots |X_{i-k:i-1})=0.$$
Hence
\begin{align}H(Y_i|X_{1:i-1})=H(Y_i|X_{i-k:i-1}).\label{eqnLmnf}\end{align}
On the other hand, from Eq. \eqref{eqn:newnewthm1} we have that conditoned on the entire input sequence, outputs at times $k+1, k+2, \cdots$ are mutually independent;  hence the following equation holds for any $k+1\leq i \leq k+r$
\begin{align*}H(Y_i|X_{i-1}X_{i-2}\cdots)&=H(Y_i|X_{i-1}X_{i-2}\cdots, X_1, Y_{i-1}, Y_{i-2}, \cdots, Y_{k+1})
\\&\leq H(Y_i|Y_{k+1:i-1}),\end{align*}
where the last equation follows from the fact that removing variables from the conditioning increases entropy. 
Therefore, \begin{align}H(Y_{k+1:k+r})=\sum_{i=k+1}^{k+r}H(Y_i|Y_{k+1:i-1})\geq \sum_{i=k+1}^{k+r}H(Y_i|X_{i-1}X_{i-2}\cdots).
\label{eqn:kjhhg}
\end{align}
Hence,
\begin{align}C_r&\geq \frac 1r I(X_{1:k+r}; Y_{k+1:k+r})\nonumber
\\&=
\frac{1}{r}\bigg[ H(Y_{k+1:k+r})-H(Y_{k+1:k+r}|X_{1:k+r})\bigg]\nonumber
\\&\geq \frac{1}{r}\bigg[\sum_{i=k+1}^{k+r}H(Y_i|X_{i-1}X_{i-2}\cdots)-H(Y_{k+1:k+r}|X_{1:k+r})\bigg]\label{eqn:mnbvc}
\\&= \frac{1}{r}\bigg[\sum_{i=k+1}^{k+r}H(Y_i|X_{i-1}X_{i-2}\cdots)-\sum_{i=k+1}^{k+r}H(Y_{i}|X_{1:k+r}Y_{k+1:i})\bigg]\label{eqn:mnbvc5}
\\&\geq \frac{1}{r}\bigg[\sum_{i=k+1}^{k+r}H(Y_i|X_{i-1}X_{i-2}\cdots)-\sum_{i=k+1}^{k+r}H(Y_{i}|X_{i-k:i})\bigg]\label{eqn:mnbvc2}
\\&=\frac{1}{r}\bigg[\sum_{i=k+1}^{k+r}H(Y_i|X_{i-k:i-1})-\sum_{i=k+1}^{k+r}H(Y_{i}|X_{i-k:i})\bigg]\label{eqn:mnbvc6}
\\&= \frac{1}{r}\sum_{i=k+1}^{k+r}I(X_i;Y_i|X_{i-k:i-1})\nonumber
\\&=I(X_{k+1};Y_{k+1}|X_{k}X_{k-1}\cdots X_1)\label{eqn:mnbvc4}
.\end{align}

where \eqref{eqn:mnbvc} follows from \eqref{eqn:kjhhg}, Eq. \eqref{eqn:mnbvc5} follows from chain rule, \eqref{eqn:mnbvc2} follows from the fact that removing variables from conditioning part can only increase entropy, \eqref{eqn:mnbvc6} follows from \eqref{eqnLmnf}, and Eq. \eqref{eqn:mnbvc4} follows from Eq. \eqref{eqn:anewa4}.

\textbf{
Proof of the upper bound:} Here we start from the upper bound $C_r$ and show that the given rate will belong to it as $r$ tends to infinity. 
Note that for any arbitrary $p(x_{1:k+r})$ we have
\begin{align}\frac 1r I(X_{1:k+r}; Y_{k+1:k+r})
&=
\frac{1}{r}\bigg[ H(Y_{k+1:k+r})-H(Y_{k+1:k+r}|X_{1:k+r})\bigg]\nonumber
\\&\leq \frac{1}{r}\bigg[\sum_{i=k+1}^{k+r}H(Y_i)-H(Y_{k+1:k+r}|X_{1:k+r})\bigg]\nonumber
\\&= \frac{1}{r}\bigg[\sum_{i=k+1}^{k+r}H(Y_i)-\sum_{i=k+1}^{k+r}H(Y_{i}|X_{i-k:i})\bigg]\label{eqn:abcdsfe2}
\\&= \frac{1}{r}\sum_{i=k+1}^{k+r}I(X_{i-k:i};Y_i)\nonumber
\\&=I(X_{Q-k:Q};Y_{Q}|Q)\nonumber
\\&\leq I(QX_{Q-k:Q};Y_{Q})\nonumber
\\&= I(X_{Q-k:Q};Y_{Q})\label{eqn:lasteqnabz}
.\end{align}
where Eq. \eqref{eqn:abcdsfe2} follows from Eq. \eqref{eqn:newnewthm1},
$Q$ is a standard time sharing variable, uniform time-sharing rv over $[k+1:k+r]$ independent of all $X_i$ and $Y_i$'s, 
and Eq. \eqref{eqn:lasteqnabz} follows from the fact that the pmf $p(Y_{q}|X_{q-k:q})$ is the same for all $q$.

We now show that the pmf of $(X_{Q-k},X_{Q-k+1},\cdots, X_{Q})$  becomes very close to one in the set $\mathcal{P}$ as $r$ goes to infinity. For arbitrary values of $x_1, \cdots, x_k$, we should consider the difference
$$\big|p(X_{Q-k}=x_1, \cdots, X_{Q-1}=x_k)-p(X_{Q-k+1}=x_1, \cdots, X_{Q}=x_k)\big|.$$ 
We have
\begin{align*}
p(X_{Q-k}=x_1, \cdots, X_{Q-1}=x_k)&=\frac{1}{r}\sum_{q=k+1}^{k+r}p(X_{q-k}=x_1, \cdots, X_{q-1}=x_k),\\
p(X_{Q-k+1}=x_1, \cdots, X_{Q}=x_k)&=\frac{1}{r}\sum_{q=k+1}^{k+r}p(X_{q-k+1}=x_1, \cdots, X_{q}=x_k).
\end{align*}
Hence, when we subtract the two, all of the terms cancel out except for two terms:
\begin{align*}\big|p&(X_{Q-k}=x_1, \cdots, X_{Q-1}=x_k)-p(X_{Q-k+1}=x_1, \cdots, X_{Q}=x_k)\big|\\&=
\frac{1}{r}\big|p(X_{1}=x_1, \cdots, X_{k}=x_k)-p(X_{r+1}=x_1, \cdots, X_{r+k}=x_k)\big|
\\&\leq \frac{1}{r}.
\end{align*}
Taking $r$ to infinity, and using the continuity of mutual information with respect to input distribution, we get the desired result. 

This completes the proof.

\section{Proof of Example \ref{example1}}
\label{AppendixProofExample1}

 Note that $$\log p(Y|X)=\log\frac{1}{\sigma \sqrt{2 \pi}} e^{\frac{-(Y-X-\mu)^2}{2\sigma ^ 2}} =\log\frac{1}{\sigma \sqrt{2 \pi}} + \frac{-(Y-X-\mu)^2}{2\sigma ^ 2},$$
and  $$\mathbb{E}_{p(x,y)}(\frac{1}{\sigma \sqrt{2 \pi}})=\mathbb{E}_{p(x)p(y)}(\frac{1}{\sigma \sqrt{2 \pi}}).$$
 For any arbitrary functions $f(X)$ and $g(Y)$ we have
 $$\mathbb{E}_{p(x,y)}f(Y)=\mathbb{E}_{p(x)p(y)}f(Y),$$
  $$\mathbb{E}_{p(x,y)}g(X)=\mathbb{E}_{p(x)p(y)}g(X).$$
Hence we get that
\begin{align*}A(p(x,y))&=\mathbb{E}_{p(x,y)}\frac{XY}{\sigma^2} - \mathbb{E}_{p(x)p(y)}\frac{XY}{\sigma^2}
\\&=\mathbb{E}_{p(x,y)}\frac{XY}{\sigma^2} - \frac{\mathbb{E}[Y]\mathbb{E}[X]}{\sigma^2}
\\&=\frac {\mathbb{E}_{p(x)}[(X+\mu)X] - \mathbb{E}[X+\mu]\mathbb{E}[X]}{\sigma^2}
\\&=\frac {\mathbb{E}_{p(x)}[X^2] - \mathbb{E}^2[X]}{\sigma^2}
\\&=\frac{\mathsf{Var}(X)}{\sigma^2}.
\end{align*}
\color{black}

\section{Proof of Example \ref{example2}}
\label{AppendixProofExample2}

 Note that $$\log p(Y|X)=\log\frac{e^{-X-\lambda_0}(X+\lambda_0)^Y}{Y!}=-(X+\lambda_0)+Y\log(X+\lambda_0)-\log(Y!).$$
Since,  $$\mathbb{E}_{p(x,y)}(X+\lambda_0)=\mathbb{E}_{p(x)p(y)}(X+\lambda_0),$$
 $$\mathbb{E}_{p(x,y)}\log(Y!)=\mathbb{E}_{p(x)p(y)}\log(Y!).$$
we get that
\begin{align*}A(p(x,y))&=\mathbb{E}_{p(x,y)}Y\log(X+\lambda_0) - \mathbb{E}_{p(x)p(y)}Y\log(X+\lambda_0)
\\&=\mathbb{E}_{p(x,y)}Y\log(X+\lambda_0) - \mathbb{E}[Y]\mathbb{E}[\log(X+\lambda_0)]
\\&=\mathbb{E}_{p(x)}(X+\lambda_0)\log(X+\lambda_0) - \mathbb{E}[X+\lambda_0]\mathbb{E}[\log(X+\lambda_0)]
\\&=\mathsf{Cov}(X+\lambda_0, \log(X+\lambda_0)).\end{align*}
A further observation is that 
$$\max_{p(x)}\mathsf{Cov}(X+\lambda_0, \log(X+\lambda_0)),$$
always occurs when $X$ is a binary random variable, whereas $\max_{p(x)}I(X;Y)$ is a non-trivial maximization over input density functions $p(x)$ with some constraints on average and maximum of $X$. To see that a binary $X$ maximizes $\mathsf{Cov}(X+\lambda_0, \log(X+\lambda_0))$, observe that this expression is equal to $\mathbb{E}[g_1(X)]$ where
$$g_1(x)=(x+\lambda_0)\log(x+\lambda_0) - (\mu+\lambda_0)(\log(x+\lambda_0)),$$
and $\mu=\mathbb{E}[g_2(X)]$ where $g_2(x)=x$. Using the Convex Cover Method of \cite[Appendix C]{ElGamalKim} with functions $g_1(x)$ and $g_2(x)$, we get that a binary $X$ suffices. \color{black} So, the solution of the problem is a discrete binary distribution. we consider these two points, $x_1, x_2$ with probabilities $p_1, p_2$. We have:
\begin{align*}
&\max_{p(x)}\mathsf{Cov}(X+\lambda_0, \log(X+\lambda_0))\\
&=\max_{\substack{p(x),\\E[X]\leq \alpha,0 \leq X \leq \mathsf{A}}}E[(X+\lambda)\log(X+\lambda)]-E[X+\lambda]E[\log(X+\lambda)]\\
&=\max_{\substack{p(x),\\E[X]\leq \alpha,0 \leq X \leq \mathsf{A}}}E[X\log(X+\lambda)]-E[X]E[\log(X+\lambda)]\color{black}\\
&=\max_{\substack{p(x),\\\E[X]=\alpha'\leq \alpha,0 \leq X \leq \mathsf{A}}}E[X\log(X+\lambda)]-\alpha'E[\log(X+\lambda)]\color{black}\\
&=\max_{\substack{p(x),\\E[X]=\alpha'\leq \alpha,0 \leq X \leq \mathsf{A}}}E[(X-\alpha')\log(X+\lambda)]\\
&=\max_{\substack{x_1p_1+x_2p_2=\alpha'\leq \alpha,\\ p_1+p_2=1,\\ 0\leq p_1,p_2,\\ 0\leq x_1,x_2 \leq \mathsf{A}}}p_1(x_1-\alpha')\log(x_1+\lambda)+p_2(x_2-\alpha')\log(x_2+\lambda)\\
&=\max_{\substack{0\leq p_1 \leq 1, 0\leq x_1 \leq \mathsf{A},\\0\leq \frac{\alpha'-p_1 x_1}{1-p_1}\leq \mathsf{A}}}p_1(x_1-\alpha')\log(\frac{x_1+\lambda}{\frac{\alpha'-p_1 x_1}{1-p_1}+\lambda}).
\end{align*}
Considering the constraints, $x_1$ constant and $0 \leq p_1 \leq \min(\frac{\alpha'-\mathsf{A}}{x_1-\mathsf{A}},\frac{\alpha'}{x_1},1)$, then if we increase $p_1$, the function $p_1(x_1-\alpha')\log(\frac{x_1+\lambda}{\frac{\alpha'-p_1 x_1}{1-p_1}+\lambda})$ also increases.  At $p_1=\min(\frac{\alpha'-\mathsf{A}}{x_1-\mathsf{A}},\frac{\alpha'}{x_1},1)$, the function is maximized.\\
We consider two cases based on the constraints:
\begin{itemize}
\item $\alpha' \leq x_1 \leq \min(\frac{\alpha'}{p_1},\mathsf{A})$ and $p_1$ is constant: The function $p_1(x_1-\alpha')\log(\frac{x_1+\lambda}{\frac{\alpha'-p_1 x_1}{1-p_1}+\lambda})$ is increasing when we increase $x_1$ for $x_1 \geq \alpha'$. So for $\alpha' \leq x_1 \leq \min(\frac{\alpha'}{p_1},\mathsf{A})$ and constant $p_1$, we have the maximum at $x_1=\min(\frac{\alpha'}{p_1},\mathsf{A})$. The maximum of the capacity  optimization is achieved at $x_1=\min(\frac{\alpha'}{p_1}, \mathsf{A})$ and $p_1=\min(\frac{\alpha'-\mathsf{A}}{x_1-\mathsf{A}},\frac{\alpha'}{x_1},1)$. So if $x_1=\mathsf{A}$ then $p_1=\frac{\alpha'}{\mathsf{A}}$ and if $x_1=\frac{\alpha'}{p_1}$ then $p_1=\frac{\alpha'-\mathsf{A}}{x_1-\mathsf{A}}$ which results in $x_1=\alpha'$ and $p_1=1$  which contradicts the two point distribution. So $x_1=\mathsf{A}$ and $x_2=0$ and $p_1=\frac{\alpha'}{\mathsf{A}}, p_2=1-\frac{\alpha'}{\mathsf{A}}$.
\item $\max(0 , \frac{\alpha'-\mathsf{A}}{p_1}+\mathsf{A})\leq x_1 \leq \alpha'$ and $p_1$ is constant: The function $p_1(x_1-\alpha')\log(\frac{x_1+\lambda}{\frac{\alpha'-p_1 x_1}{1-p_1}+\lambda})$ is increasing when we decrease $x_1$ for $\max(0 , \frac{\alpha'-\mathsf{A}}{p_1}+\mathsf{A})\leq x_1 \leq \alpha'$. We have the maximum at $x_1=\max(0 , \frac{\alpha'-\mathsf{A}}{p_1}+\mathsf{A})$. The maximum of the capacity optimization is achieved at $x_1=\max(0 , \frac{\alpha'-\mathsf{A}}{p_1}+\mathsf{A})$ and $p_1=\min(\frac{\alpha'-\mathsf{A}}{x_1-\mathsf{A}},\frac{\alpha'}{x_1},1)$. Now if $x_1=0$ then $p_1=1-\frac{\alpha'}{\mathsf{A}}$ and if $x_1=\frac{\alpha'-\mathsf{A}}{p_1}+\mathsf{A}$ then $p_1=\frac{\alpha'-\mathsf{A}}{x_1-\mathsf{A}}$ which results in $x_1=\alpha'$ and $p_1=1$ which also results in a contradiction of the needed binary valued distribution.
\end{itemize}

Thus, the optimal distribution is $p(x)=\frac{\alpha'}{\mathsf{A}}\delta(x-\mathsf{A})+(1-\frac{\alpha'}{\mathsf{A}}) \delta(x)$ and the upper bound is $$\max_{\alpha'\leq \alpha}\frac{\alpha'}{\mathsf{A}}(\mathsf{A}-\alpha')\log(\frac{\mathsf{A}}{\lambda_0}+1),$$
which is equal to $\frac{\alpha}{\mathsf{A}}(\mathsf{A}-\alpha)\log(\frac{\mathsf{A}}{\lambda_0}+1)$ if $\alpha\leq \mathsf{A}/2$, and 
$\frac{\mathsf{A}}{4}\log(\frac{\mathsf{A}}{\lambda_0}+1)$ otherwise.

\section{Properties of the Upper Bound (Section \ref{section:someproperties})}
\label{AppendixPropertiesProof}
\emph{Proof of Property 2:} This follows from algebra:
\begin{align*}&A\big(p(y_1|x_1)p(y_2|x_2)\big)\\&=\max_{p(x_1, x_2)}\mathbb{E}_{p(x_1x_2y_1y_2)}\log p(Y_1Y_2|X_1X_2) - \mathbb{E}_{p(x_1x_2)p(y_1y_2)}\log p(Y_1Y_2|X_1X_2) 
\\&=\max_{p(x_1, x_2)}\mathbb{E}_{p(x_1x_2y_1y_2)}\log p(Y_1|X_1)p(Y_2|X_2) - \mathbb{E}_{p(x_1x_2)p(y_1y_2)}\log p(Y_1|X_1)p(Y_2|X_2) 
\\&=\max_{p(x_1)}\mathbb{E}_{p(x_1y_1)}\log p(Y_1|X_1) - \mathbb{E}_{p(x_1)p(y_1)}\log p(Y_1|X_1)
\\&\qquad+\max_{p(x_2)}\mathbb{E}_{p(x_2y_2)}\log p(Y_2|X_2) - \mathbb{E}_{p(x_2)p(y_2)}\log p(Y_2|X_2).
\end{align*}

\emph{Proof of Property 3:} Using Jensen's inequality as follows, we have
\begin{align*}H(Y)&=-\sum_{y}p(y)\log(p(y))
\\&= -\sum_{y}p(y)\log(\sum_xp(x)p(y|x))
\\&\leq -\sum_{y}p(y)\sum_xp(x)\log(p(y|x))
\\&=-\sum_{x,y}p(x)p(y)\log(p(y|x)).
\end{align*}
Thus, \begin{align*}I(X;Y)&=H(Y)-H(Y|X)\\&\leq -\sum_{x,y}p(x)p(y)\log(p(y|x)) - H(Y|X)
\\&= -\sum_{x,y}p(x)p(y)\log(p(y|x)) + \sum_{x,y} p(x,y)\log p(y|x)\color{black}
\\&=\sum_{x,y}\big[p(x,y)-p(x)p(y)\big]\log(p(y|x)).\end{align*}

\section{Proof of Theorem \ref{prop2}}
\label{AppendixProofClaim2}

A code of length $n$ consists of a set of codewords $X^n(m), m\in [1:2^{nR}]$. The codeword $X^n(m)$ corresponds to a sequence $S^n(m)$ after passing through the LTI system. The sequences $S^n(m), m\in [1:2^{nR}]$ can then be thought of as codewords for a memoryless Poisson channel. The sequence $S^n(m)$ are the outputs of the LTI system and satisfy some linear constraints. Let 
\begin{align*}\mathcal{S}_n(\mathsf{A}, \alpha, \textbf{p})=\{\textbf{s}=(s_1, s_2, \cdots, s_n): \exists \textbf{x}=(x_1, x_2, \cdots, x_n):~ & \textbf{s}=(\textbf{x}{\star} \textbf{p})_{n}, \mathbf{x}\in\mathcal{P}_n(\mathsf{A}, \alpha)\},
\end{align*}
where 
$$\mathcal{P}_n(\mathsf{A}, \alpha)=\{(x_1, x_2, \cdots,x_n): x_i\geq 0, ~\frac{1}{n}\sum_{i=1}^nx_i\leq \alpha,~ x_i\leq \mathsf{A}\},$$
and by $(\textbf{x}\star‎ \textbf{p})_n$ we mean truncated convolution, i.e. that the first $n$ elements  of the convolution is taken (the convolution can have more terms).

Therefore, achieving the capacity $\mathcal{C}(\mathsf{A}, \alpha, \textbf{p}, \lambda_0)$ is equivalent to choosing the best possible codewords from the set  $\mathcal{S}_n(\mathsf{A}, \alpha, \textbf{p})$. 
It suffices to show that for each $n$, 
 $\mathcal{S}_n(\mathsf{A}, \alpha, \textbf{p}')\subset \mathcal{S}_n(\mathsf{A}, \alpha, \textbf{p})$. This implies that there is more freedom to choose the codewords in the problem with $\textbf{p}$ than in the problem with $\textbf{p}'$.  Select an arbitrary  $\mathbf{s}\in \mathcal{S}_n(\mathsf{A}, \alpha, \textbf{p}')$. We would like to show that $\mathbf{s}\in \mathcal{S}_n(\mathsf{A}, \alpha, \textbf{p})$. Since $\mathbf{s}\in \mathcal{S}_n(\mathsf{A}, \alpha, \textbf{p}')$, we have,
$$\textbf{s}=(\textbf{x}{\star} \textbf{p}')_n=(\textbf{x}{\star}  (\textbf{p}{\star} \textbf{q}))_n=((\textbf{x}{\star} \textbf{q}){\star} \textbf{p})_n=((\textbf{x}{\star} \textbf{q})_n{\star} \textbf{p})_n,$$
for some $\textbf{x}\in \mathcal{P}_n(\mathsf{A}, \alpha)$. It suffices to show that $\textbf{r}:=(\textbf{x}{\star} \textbf{q})_n$ is in $\mathcal{P}_n(\mathsf{A}, \alpha)$ to conclude that $\mathbf{s}\in \mathcal{S}_n(\mathsf{A}, \alpha, \textbf{p})$.We have
\begin{align*}\frac{1}{n}\sum_{i=1}^nr_i&=\frac{1}{n}\sum_{i=1}^n\sum_{j=0}^ix_jq_{i-j}
\\&=\frac{1}{n}\sum_{i=1}^nx_i\sum_{j=0}^{n-i}q_{j}
\\&\leq \frac{1}{n}\sum_{i=1}^nx_i(\sum_{i=0}^{n-1}q_i)
\\&\leq \alpha(\sum_{i=0}^{n-1}q_i) 
\\&\leq \alpha,
\end{align*}
where in the last step we use the fact that $\sum_{i=0}^{\infty}q_i=1$ (this is because $\textbf{p}'=\textbf{p}{\star} \textbf{q}$ and both $\textbf{p}$ and $\textbf{p}'$ have elements that sum to one). Therefore, $\sum_{i=1}^nr_i/N\leq \alpha$. On the other hand,
\begin{align*}r_i&=\sum_{j=0}^ix_jq_{i-j}
\leq \mathsf{A}\sum_{j=0}^iq_{i-j}
\leq \mathsf{A}(\sum_{i=0}^{n-1}q_i)
\leq \mathsf{A}.
\end{align*}
Therefore $\textbf{r}=(\textbf{x}{\star} \textbf{q})_n\in\mathcal{P}_n(\mathsf{A}, \alpha)$. This completes the proof.
\color{black}


\end{document}